\documentclass[11pt]{article}
\usepackage{sustyle}

\usepackage[utf8]{inputenc} 
\usepackage[T1]{fontenc}    
\usepackage{lmodern}

\usepackage{hyperref}       
\usepackage{url}            
\usepackage{booktabs}       
\usepackage{amsfonts}       
\usepackage{nicefrac}       
\usepackage{microtype}      
\usepackage{amsthm}
\usepackage{tikz} 
\usepackage{float}
\usepackage{amsmath, amssymb}
\usepackage{array}
\usepackage{diagbox}
\usepackage{mathtools}
\usepackage{multirow}
\usepackage{rotating}
\usepackage{cleveref}
\usepackage[most]{tcolorbox}
\usepackage{bigints}

\usepackage{subcaption}
\usepackage{siunitx}
\usepackage{authblk}

\definecolor{ired}{rgb}{0.9,0,0.1}

\newtcolorbox{Box1}[2][]{
	lower separated=false,
	colback=white!80!gray,
	colframe=white, fonttitle=\bfseries,
	colbacktitle=white!50!gray,
	coltitle=black,
	enhanced,
	attach boxed title to top left={xshift=0.5cm,yshift=-5mm},
	title=#2,#1}

\usepackage[round]{natbib}

\renewcommand{\l}{\left}
\renewcommand{\r}{\right}
\renewcommand{\E}{\mathbb{E}}

\renewcommand{\R}{\mathbb{R}}
\newcommand{\K}{\mathbb{K}}
\renewcommand{\P}{\mathbb{P}}

\newcommand{\W}{\mathcal{W}}
\newcommand{\X}{\mathcal{X}}
\renewcommand{\A}{\mathcal{A}}
\newcommand{\M}{\mathcal{M}}
\renewcommand{\L}{\mathcal{L}}

\graphicspath{{Figs/}}

\title{Privacy Amplification via Iteration for Shuffled \\ and Online PNSGD}
\date{}

\author[1]{Matteo Sordello}
\author[2]{Zhiqi Bu}
\author[3]{Jinshuo Dong}
\affil[1]{Department of Statistics, University of Pennsylvania}
\affil[2]{Graduate Group in AMCS, University of Pennsylvania}
\affil[3]{IDEAL Institute, Northwestern University}

{
	\makeatletter
	\renewcommand\AB@affilsepx{: \protect\Affilfont}
	\makeatother

	\affil[1]{\texttt{sordello@wharton.upenn.edu}}
	\affil[2]{\texttt{zbu@sas.upenn.edu}}
	\affil[3]{\texttt{jinshuo@northwestern.edu}}
}

\begin{document}

\maketitle

\begin{abstract}
In this paper, we consider the framework of privacy amplification via iteration, which is originally proposed by Feldman et~al.~and subsequently simplified by Asoodeh et~al.~in their analysis via the contraction coefficient.
This line of work focuses on the study of the privacy guarantees obtained by the projected noisy stochastic gradient descent (PNSGD) algorithm with hidden intermediate updates.
A limitation in the existing literature is that only the early stopped PNSGD has been studied, while no result has been proved on the more widely-used PNSGD applied on a shuffled dataset.
Moreover, no scheme has been yet proposed regarding how to decrease the injected noise when new data are received in an online fashion.
In this work, we first prove a privacy guarantee for shuffled PNSGD, which is investigated asymptotically when the noise is fixed for each sample size $n$ but reduced at a predetermined rate when $n$ increases, in order to achieve the convergence  of privacy loss. We then analyze the online setting and provide a faster decaying scheme for the magnitude of the injected noise that also guarantees the convergence of privacy loss.
\end{abstract}

\section{Introduction}
\label{sec_intro}

Differential privacy (DP) \citet{dwork2006calibrating}, \citet{dwork2006our} is a strong standard to guarantee the privacy for algorithms that have been widely applied to modern machine learning \citep{abadi2016deep}. It characterizes the privacy loss via statistical hypothesis testing, thus allowing the mathematically rigorous analysis of the privacy bounds. 
When multiple operations on the data are involved
 and each intermediate step is revealed, composition theorems can be used to keep track of the privacy loss, which combines subadditively \citep{kairouz2015composition}. However, because such results are required to be general, their associated privacy bounds are inevitably loose. In contrast, privacy amplification provides a privacy budget for a composition of mechanisms that is less that the budget of each individual operation, which strengthens the bound the more operations are concatenated. Classic examples of this feature are privacy amplification by subsampling \citep{chaudhuri2006random, balle2018privacy}, by shuffling \citep{erlingsson2019amplification} and by iteration \citep{feldman2018privacy, asoodeh2020privacy}.
In this paper, we focus on the setting of privacy amplification by iteration, and extend the analysis via contraction coefficient proposed by \citet{asoodeh2020privacy} to prove results that apply to an algorithm commonly used in practice, in which the entire dataset is shuffled before training a model with PNSGD.
We emphasize that the shuffling is a fundamental difference compared to previous work, since it is a necessary step in training many machine learning models.

We start by laying out the definitions that are necessary for our analysis.
We consider a convex function $f : \R^+ \to \R$ that satisfies $f(1) = 0$. \citet{ali1966general} and \citet{csiszar2004information} define the $f$-divergence between two probability distribution $\mu$ and $\nu$ is as
\begin{equation*}
\label{definition_f_divergence}
D_f(\mu \| \nu) = \E_{\nu}\l[f\l(\frac{\text{d}\mu}{\text{d}\nu}\r)\r]=\int f\l(\frac{\text{d}\mu}{\text{d}\nu}\r)d\nu
\end{equation*}
For a Markov kernel $K: \W \to \mathcal{P}(\W)$, where $\mathcal{P}(\W)$ is the space of probability measures over $\W$, we let $\eta_f(K)$ be the contraction coefficient of kernel $K$ under the $f$-divergence, which is defined as
\begin{equation*}
\label{definition_contraction_coefficient}
\eta_f(K) = \sup_{\mu, \nu : D_f(\mu \| \nu) \neq 0} \frac{D_f(\mu K \| \nu K)}{D_f(\mu \| \nu)}
\end{equation*}
If we now consider a sequence of Markov kernels $\{K_n\}$ and let the two sequences of measures $\{\mu_n\}$ and $\{\nu_n\}$ be generated starting from $\mu_0$ and $\nu_0$ by applying $\mu_n = \mu_{n-1} K_{n}$ and $\nu_n = \nu_{n-1} K_{n}$, then the strong data processing inequality \citep{raginsky2016strong} for the $f$-divergence tells us that
\begin{equation*}
\label{strong_data_processing_inequality}
D_f(\mu_n \| \nu_n) \leq D_f(\mu_0 \| \nu_0) \prod_{t=1}^{n} \eta_f(K_t)
\end{equation*}
Among the $f$-divergences, we focus on the $E_\gamma$-divergence, or hockey-stick divergence, which is the $f$-divergence associated with $f(t) = (t - \gamma)_+=\max(0,t-\gamma)$. We do so because of its nice connection with the concept of $(\epsilon, \delta)$ differential privacy, which is now the state-of-the-art technique to analyze the privacy loss that we incur when releasing information from a dataset.  A mechanism $\M$ is said to be $(\epsilon, \delta)$-DP if, for every pair of neighboring datasets (datasets that differ only in one entry, for which we write $D \sim D'$) and every event $\A$, one has
\begin{equation}
\label{definition_DP}
\P(\M(D) \in \A) \leq e^\epsilon \P(\M(D') \in \A) + \delta
\end{equation}
It is easy to prove that a mechanism $\M$ is $(\epsilon, \delta)$-DP if and only if the distributions that it generates on $D$ and $D'$ are close with respect to the $E_\gamma$-divergence. In particular, for $D \sim D'$ and $\P_D$ being the output distribution of mechanism $\M$ on $D$, then $\M$ is $(\epsilon, \delta)$-DP if and only if 
\begin{equation}
\label{connection_E_divergence_privacy}
E_{e^\epsilon}(\P_D \| \P_{D'}) \leq \delta .
\end{equation}
It has been proved in \citet{asoodeh2020privacy} that the contraction coefficient of a kernel $K: \W \to \mathcal{P}(\W)$ under $E_\gamma$-divergence, which we refer to as $\eta_\gamma(K)$, satisfies
\begin{equation*}
\label{contraction_coefficient_gamma_divergence}
\eta_\gamma(K) = \sup_{w_1, w_2 \in \W} E_\gamma(K(w_1) \| K(w_2))
\end{equation*}
This equality improves on a result proved by \citet{balle2019privacy} and makes it easier to find an explicit form for the contraction coefficient of those distributions for which we can compute the hockey-stick divergence. Two such distributions are the Laplace and Gaussian, and \citet{asoodeh2020privacy} investigate the privacy guarantees generated by this privacy amplification mechanism in the setting of PNSGD with Laplace or Gaussian noise. 
As the standard stochastic gradient descent (SGD), the PNSGD is defined with respect to a loss function $\ell : \W \times \X \to \R$ that takes as inputs a parameter in the space $\K \subseteq \W$ and an observation $x \in \X$. Common assumptions made on the loss functions are the following: for each $x \in \X$
\begin{itemize}
\item $\ell(\cdot, x)$ is $L$-Lipschitz \\[-8pt]
\item $\ell(\cdot, x)$ is $\rho$-strongly convex \\[-8pt]
\item $\nabla_w \ell(\cdot, x)$ is $\beta$-Lipschitz. 
\end{itemize}
The PNSGD algorithm works by combining three steps:  $(1)$ a stochastic gradient descent (SGD) step with learning rate $\eta$; $(2)$ an injection of i.i.d. noise sampled from a known distribution to guarantee privacy and $(3)$ a projection $\Pi_{\K}: \W \to \K$ onto the subspace $\K$. Combined, these steps give the following update rule
\begin{Box1}{PNSGD}
\begin{equation*}
\label{PNSGD_updates}
w_{t+1} = \Pi_{\K}\l(w_t - \eta (\nabla_w\ell(w_t, x_{t+1}) + Z_{t+1})\r)
\end{equation*}
\end{Box1}
which can be defined as a Markov kernel by assuming that $w_0 \sim \mu_0$ and $w_{t} \sim \mu_t = \mu_0 K_{x_1} ... K_{x_t}$, where $K_x$ is the kernel associated to a single PNSGD step when observing the data point $x$. With this definition, one can find an upper bound for $\delta$ by bounding the left hand side of (\ref{connection_E_divergence_privacy}). The specific bound depends on the index at which the neighboring datasets $D$ and $D'$ differ and the distribution of the noise injected in the PNSGD. \citet{asoodeh2020privacy} investigate the bound for both Laplace and Gaussian noise, which we report in the following theorem.
\begin{theorem}[Theorem 3 and 4 in \citet{asoodeh2020privacy}]
\label{thm:harvard}
Define
\begin{equation*}
\label{definition_function_Q}
Q(t) = \frac{1}{\sqrt{2\pi}} \int_t^{\infty} e^{-\frac{u^2}{2}} du = 1-\Phi(t)
\end{equation*}
where $\Phi$ is the cumulative density function of the standard normal,
\begin{equation}
\label{definition_function_theta}
\theta_\gamma(r) = Q\l(\frac{\log(\gamma)}{r} - \frac r2\r) - \gamma Q\l(\frac{\log(\gamma)}{r} + \frac r2\r)
\end{equation}
and the constant
\begin{equation*}
\label{definition_function_M}
M = \sqrt{1 - \frac{2\eta\beta\rho}{\beta + \rho}}
\end{equation*}
which depends on the parameters of the loss function and the learning rate of the SGD step.
If $\K \subset \R^d$ is compact and convex with diameter $D_\K$, the PNSGD algorithm with Gaussian noise $N(0, \sigma^2)$ is $(\epsilon, \delta)$-DP for its i-th entry where $\epsilon \geq 0$ and 
\begin{equation*}
\label{delta_PNSGD_gaussian}
\delta = \theta_{e^\epsilon}\l(\frac{2L}{\sigma}\r)  \theta_{e^\epsilon}\l(\frac{M D_{\mathbb{K}}}{\eta \sigma}\r)^{n-i}
\end{equation*}
If instead we consider $\K = [a, b]$ for $a < b$, then the PNSGD algorithm with Laplace noise $\L(0, v)$ is $(\epsilon, \delta)$-DP for its i-th entry where $\epsilon \geq 0$ and 
\begin{equation*}
\label{delta_PNSGD_laplace}
\delta = \l(1 - e^{\frac{\epsilon}{2} - \frac Lv}\r)_+ \l(1 - e^{\frac \epsilon2 - \frac{M(b-a)}{2\eta v}}\r)_+^{n-i}
\end{equation*}
\end{theorem}

To slightly simplify the notation, we can present the guarantees in \Cref{thm:harvard} as $\delta=A\cdot B^{n-i}$ where for the Gaussian case
\begin{equation}
\label{A_B_gaussian}
	A = \theta_{e^\epsilon}\l(\frac{2L}{\sigma}\r) \quad\text{and}\quad B = \theta_{e^\epsilon}\l(\frac{M D_{\mathbb{K}}}{\eta \sigma}\r)
\end{equation}
and for the Laplacian case
\begin{equation}
\label{A_B_laplace}
	A = \l(1 - e^{\frac{\epsilon}{2} - \frac Lv}\r)_+ \quad\text{and}\quad B = \l(1 - e^{\frac \epsilon2 - \frac{M(b-a)}{2\eta v}}\r)_+
\end{equation}

To get a bound that does not depend on the index of the entry on which the two datasets differ, the authors later consider the randomly-stopped PNSGD, which simply consist of picking a random stopping time for the PNSGD uniformly from $\{1, ..., n\}$.
The bound that they obtain for $\delta$ in the Gaussian case is $\delta = A/[n(1-B)]$. Based on their proof, it is clear that the actual bound contains a term $(1 - B^{n-i+1})$ at the numerator and that the same result can be obtained if we consider the Laplace noise.

In Section \ref{sec_shuffling} we prove that a better bound than the one obtained via randomly-stopped PNSGD can be obtained by first shuffling the dataset and then applying the simple PNSGD. In Section \ref{sec_asymptotic} we study the asymptotic behavior of such bound and find the appropriate decay rate for the variability of the noise level that guarantees convergence for $\delta$ to a non-zero constant.

\section{Related Work}
\label{sec_related_work}

In the DP regime, $(\epsilon,\delta)$-DP (see \eqref{definition_DP}) is arguably the most popular definition, which is oftentimes achieved by an algorithm which contains Gaussian or Laplacian noises. For example, in NoisySGD and NoisyAdam in \citet{abadi2016deep,bu2020deep}, and PNSGD in this paper, a certain level of random noise is injected into the gradient to achieve DP. Notably, as we use more datapoints (or more iterations during the optimization) during the training procedure, the privacy loss accumulates at a rate that depends on the magnitude of the noise.

It is remarkably important to charaterize, as tightly as possible, the privacy loss at each iteration. An increasing line of works have proposed to address this difficulty \citep{dong2019gaussian,bun2016concentrated,dwork2016concentrated,balle2018privacy,mironov2017renyi,wang2019subsampled,koskela2020computing,asoodeh2020privacy,abadi2016deep}, which bring up many useful notions of DP, such as R\'enyi DP, Gaussian DP, $f$-DP and so on. Our paper extends \citet{asoodeh2020privacy} by shuffling the dataset first rather than randomly stopping the PNSGD (see Theorem 5 in \citet{asoodeh2020privacy}), in order to address the non-uniformity of privacy guarantee. As a consequence, we obtain a strictly better privacy bound and better loss than \citet{asoodeh2020privacy}, \citet{abadi2016deep}, and an additional online result of the privacy guarantee.

Furthermore, our results can be easily combined with composition tools in DP \citep{kairouz2015composition, abadi2016deep, koskela2020computing, dong2019gaussian}. In \Cref{thm_shuffled_pnsgd}, \Cref{thm_shuffled_fixed_laplace} and \Cref{thm_shuffled_fixed_gaussian}, the $(\epsilon,\delta)$ is computed based on a single pass of the entire dataset, or equivalently on one epoch. When using the shuffled PNSGD for multiple epochs, as is usual for modern machine learning, the privacy loss accumulates and is accountable by Moments accountant (using Renyi DP \citep{mironov2017renyi}), $f$-DP (using functional characterization of the type I/II errors trade-off) and other divergence approaches.

\section{Shuffled PNSGD}
\label{sec_shuffling}

In this section, we prove the bound on $\delta$ that we can obtain by first shuffling the dataset and then apply the PNSGD algorithm. The simple underlying idea here is that, when shuffling the dataset, the index at which the two neighboring datasets differ has equal probability to end up in each position. This is a key difference compared to the randomly-stopped PNSGD, and allows us to get a better bound that do not depend on the initial position of that index.

\begin{theorem}
\label{thm_shuffled_pnsgd}
Let $D \sim D'$ be of size $n$. Then the shuffled PNSGD is $(\epsilon, \delta)$-DP with 
\begin{equation}
\label{delta_shuffled_PNSGD}
\delta = \frac{A\cdot (1 - B^n)}{n(1-B)} 
\end{equation}
and the constants $A$ and $B$ are defined in (\ref{A_B_gaussian}) for Gaussian noise and (\ref{A_B_laplace}) for Laplace noise.
\end{theorem}
\begin{proof}
Let's start by considering the simple case $n = 2$, so that $D = \{x_1, x_2\}$ and $D' = \{x_1', x_2'\}$ and let $i \in \{1, 2\}$ be the index at which they differ. Let $\mu$ be the output distribution of the shuffled PNSGD on $D$, and $\nu$ be the corresponding distribution from $D'$. If we define $S(D)$ and $S(D')$ to be the two dataset after performing the same shuffling, then we can only have either $S(D) = \{x_1, x_2\}$ or $S(D) = \{x_2, x_1\}$, both with equal probability $1/2$. The outcomes of the shuffled PNSGD on $D$ and $D'$ are then
\begin{align*}
\mu &= \frac12 \mu_0 K_{x_1} K_{x_2} + \frac12 \mu_0 K_{x_2} K_{x_1} \\[3pt]
\nu &= \frac12 \mu_0 K_{x_1'} K_{x_2'} + \frac12 \mu_0 K_{x_2'} K_{x_1'}
\end{align*}
By convexity and Jensen's inequality we have that
\begin{align*}
E_\gamma(\mu \| v) &\leq \frac12 E_\gamma \l(\mu_0 K_{x_1} K_{x_2} \| \mu_0 K_{x_1'} K_{x_2'}\r) \\
&\quad + \frac12 E_\gamma \l(\mu_0 K_{x_2} K_{x_1} \| \mu_0 K_{x_2'} K_{x_1'}\r)
\end{align*}
and now we have two options, based on where the two original datasets differ. If $i = 1$, in the first term the privacy is stronger than in the second one (because $x_1$ is seen earlier), and we have 
\begin{equation*}
E_\gamma(\mu \| \nu) \leq \frac12 A\cdot B + \frac12 A = \frac12 A (B + 1)
\end{equation*}
If $i = 2$, now the privacy is stronger in the second term, and
\begin{equation*}
E_\gamma(\mu \| \nu) \leq \frac12 A + \frac12 A\cdot B = \frac12 A (B + 1)
\end{equation*}
Since in both cases the bound is the same, this means that for any $i \in \{1, 2\}$ the privacy guarantee of the shuffled PNSGD algorithm is equal to $A (B + 1)/2$.
From here we see that, when $n > 2$, the situation is similar. Instead of just two, we have $n!$ possible permutations for the elements of $D$, each one happening with the same probability $1/n!$. 
For each fixed index $i$ on which the two neighboring datasets differ, we have $(n-1)!$ permutations in which element $x_i$ appears in each of the $n$ positions. When, after the permutation, element $x_i$ ends up in last position, the bound on $E_\gamma(\mu \| \nu)$ is the weakest and just equals $A$. When in ends up in first position, the bound is the strongest and is equal to $A\cdot B^{n-1}$.
We then have that, irrespectively of the index $i$,
\begin{equation*}
E_\gamma(\mu \| v) \leq \frac{1}{n!} (n-1)! A \sum_{j=0}^{n-1} B^j = \frac{A\cdot (1 - B^n)}{n(1-B)}
\end{equation*}
\end{proof}
This bound is indeed better than the one found in \citet{asoodeh2020privacy} for the randomly stopped PNSGD since it contains an extra term $(1 - B^n)$ at the numerator which does not depend on $i$ and is smaller than $1$. If $n$ is large and $B$ is fixed, this difference is negligible because it decays exponentially. However, we will see later that when the injected noise is reduced at the appropriate rate we can guarantee that $B \approx 1 - O(1/n)$, so that the extra term ends up having an impact in the final bound. 
It is also important to notice that shuffled PNSGD achieves in general better performance than randomly stopped PNSGD and it is much more commonly used in practice. We see in \Cref{fig_shuffle_vs_early_stop} that this is the case for both linear and logistic regression, and that the variation in the result in shuffled PNSGD is less than for the early stopped case, due to the fact that we always use all the data available for each epoch.
In the next section we look at the asymptotic behavior of (\ref{delta_shuffled_PNSGD}) when $n$ grows and the variance of the injected noise is properly reduced to guarantee convergence.

\begin{figure}[!thbp]
  \centering
  \begin{minipage}{0.45\linewidth}
  \centering
  \includegraphics[width=\linewidth]{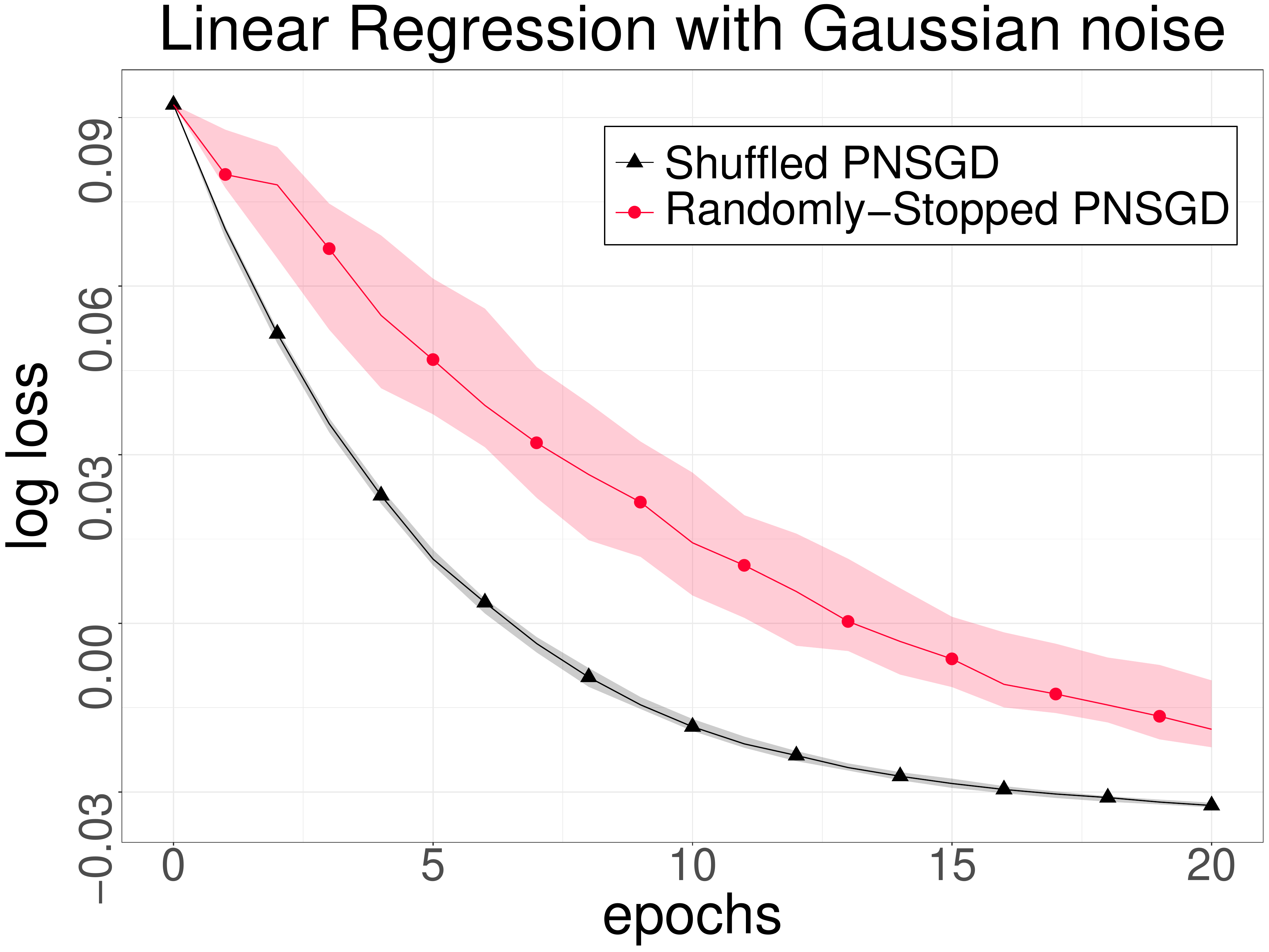}
  \end{minipage}
  \begin{minipage}{0.45\linewidth}
  \centering
  \includegraphics[width=\linewidth]{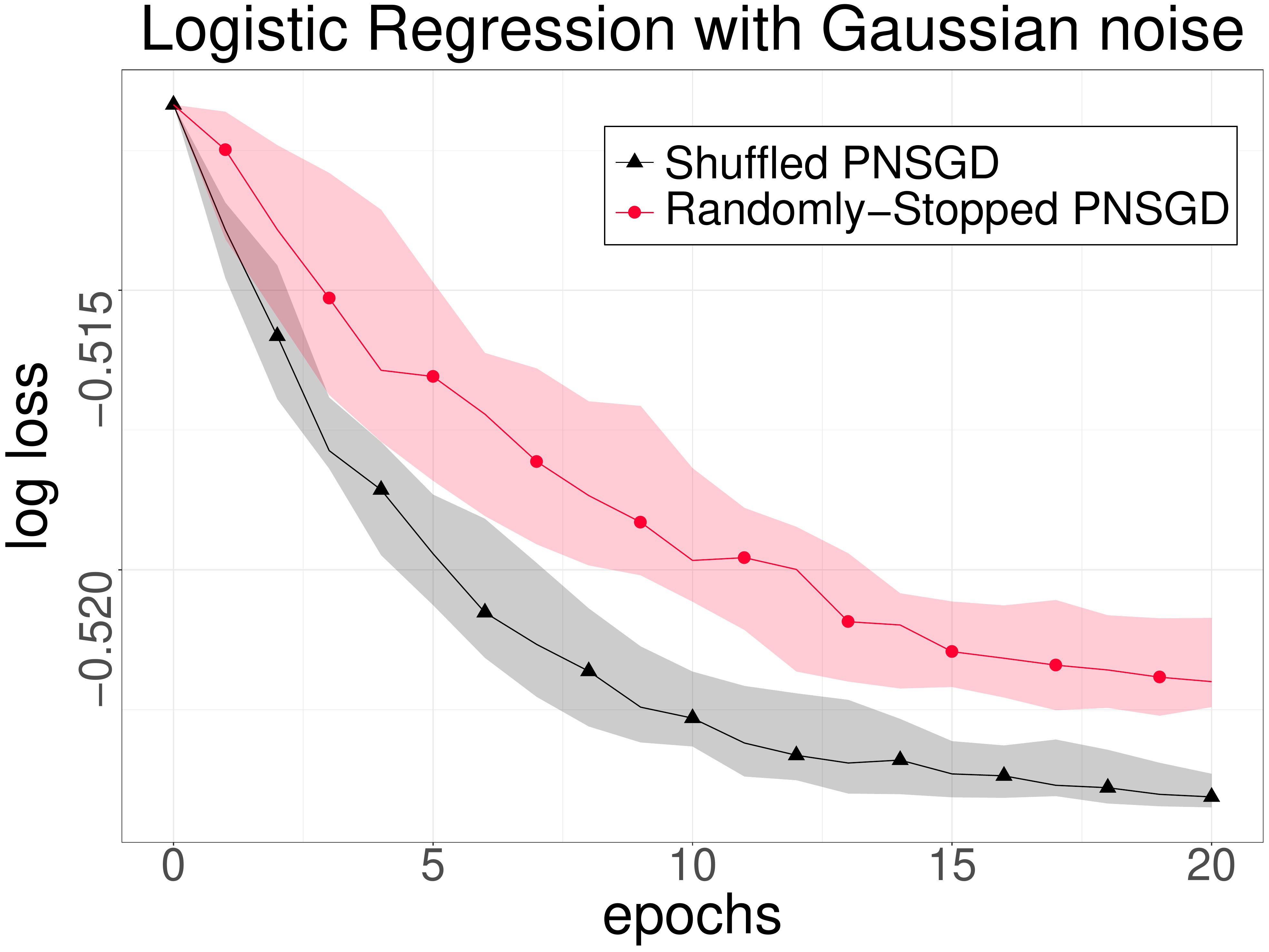}
  \end{minipage}  
  \caption{Comparison between shuffled PNSGD and randomly-stopped PNSGD with Gaussian noise in linear and logistic regression. On the y-axis we report the log loss achieved. The parameters used are $n = 1000, d=2, \sigma=0.5, \theta^* = \Pi_{\K}(1,2)$ and $\K$ is a ball of radius $1$. The learning rate is $10^{-4}$ in linear regression and $5\cdot 10^{-3}$ in logistic regression.}
\label{fig_shuffle_vs_early_stop}
\end{figure}


\section{Asymptotic Analysis for $\delta$ when Using Shuffling and Fixed Noises}
\label{sec_asymptotic}

In this Section we investigate the behavior of the differential privacy bound in (\ref{delta_shuffled_PNSGD}) when the size $n$ of the dataset grows. In Section \ref{subsec_asymptotic_laplace} we prove a results for the shuffled PNSGD with fixed Laplace noise, while in Section \ref{subsec_asymptotic_gaussian} we prove the same result on the shuffled PNSGD with fixed Gaussian noise.

\subsection{Laplace Noise}
\label{subsec_asymptotic_laplace}

We present first a result that holds when we consider a fixed Laplace noise $\L(0, v)$ injected into the PNSGD algorithm for each update. In order to get a convergence result for $\delta$ as the size $n$ of the dataset grows, the level of noise that we use should be targeted to the quantity $n$.  The decay of $v$ is regulated by two parameters, $C_1$ and $C_2$. While $C_1$ is set to be large, so that $\delta$ converges to a small value, the use of $C_2$ is simply to allow the noise level not to be too large for small $n$, but does not appear in the asymptotic bound.

\begin{figure}[tb]
  \centering
  \begin{minipage}{0.49\linewidth}
  \centering
  \includegraphics[width=\linewidth]{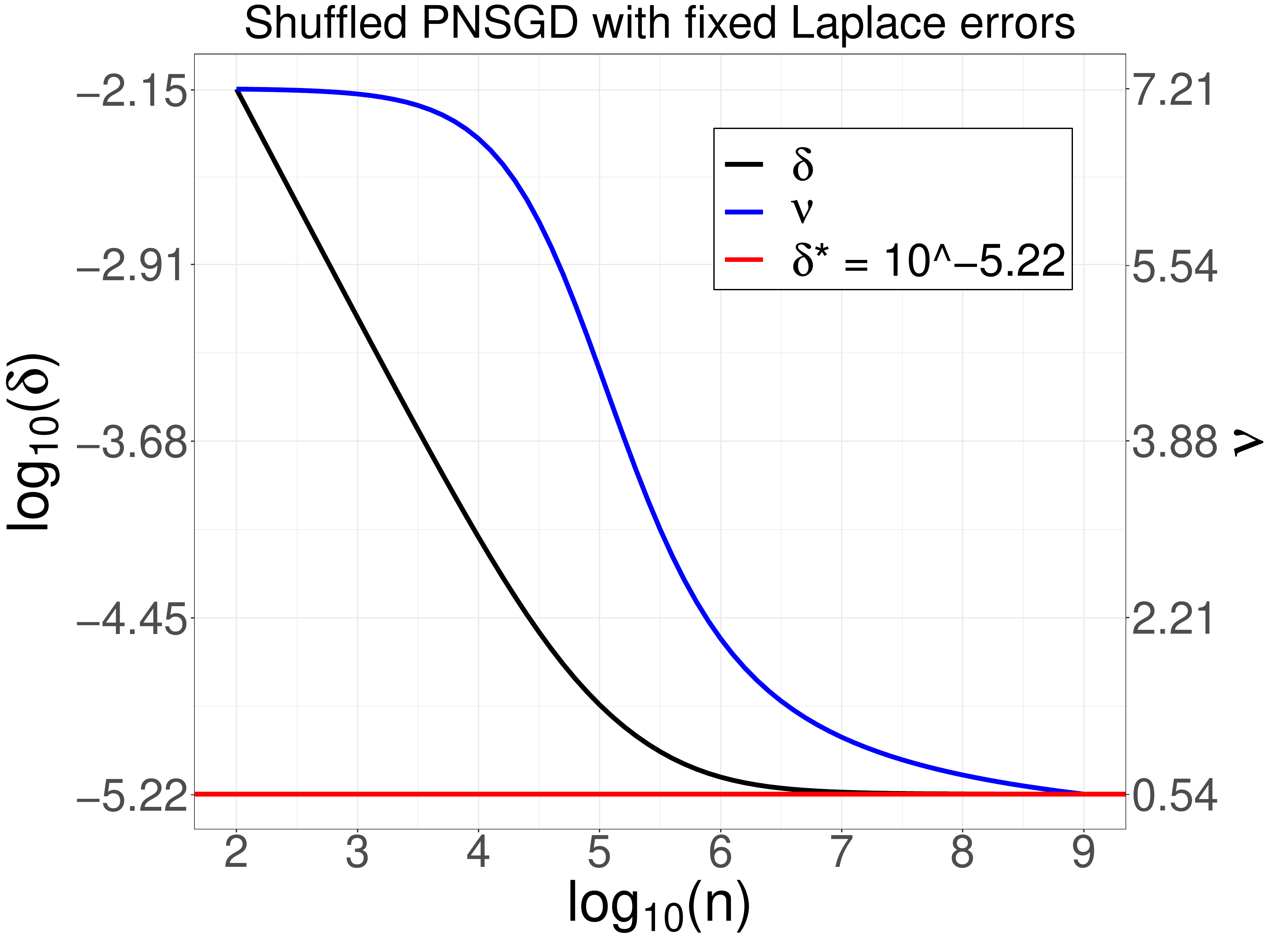}
  \end{minipage}
  \begin{minipage}{0.49\linewidth}
  \centering
  \includegraphics[width=\linewidth]{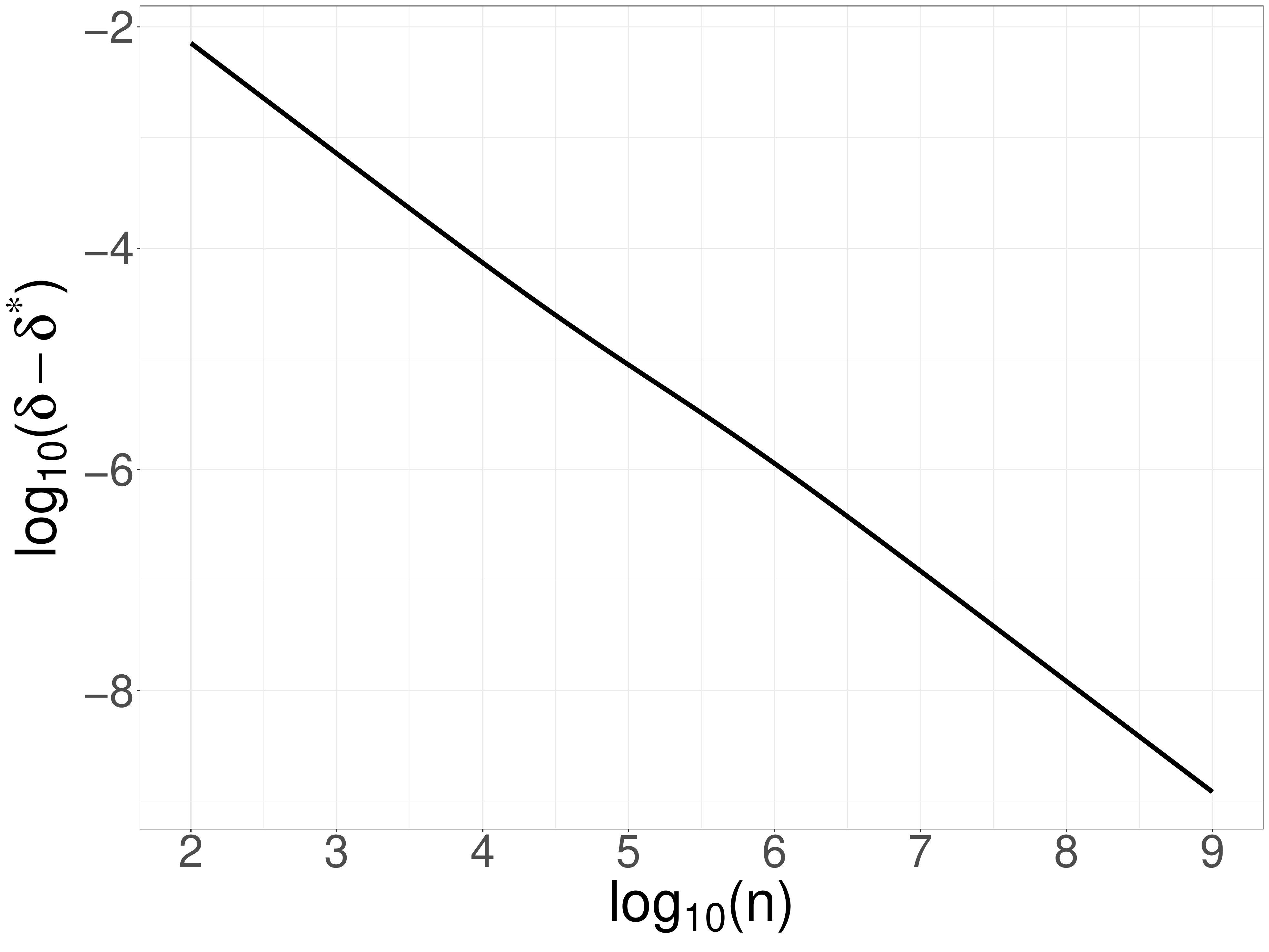}
  \end{minipage}  
  \caption{(left) Convergence of $\delta$ to $\delta^*$ in (\ref{delta*_laplace_shuffled}). We plot in black the behavior of $\delta$ as a function of $n$, and in blue the corresponding behavior of $v(n)$ in (\ref{v_laplace_shuffled}). (right) We show that the convergence rate is $1/n$. The parameters used are $L=10, \beta=0.5, \rho=0, \eta = 0.1, \epsilon=1, (a,b)=(0,1), C_1=10^5$ and $C_2=2$.}
\label{fig_delta_shuffled_fixed_laplace}
\end{figure}


\begin{theorem}
\label{thm_shuffled_fixed_laplace}
Consider the shuffled PNSGD with Laplace noise $\L(0, v(n))$ which is fixed for each update, where 
\begin{equation}
\label{v_laplace_shuffled}
v(n) = \frac{M(b-a)}{2\eta \log\l(n/C_1 + C_2\r)} .
\end{equation}
Then, for $n$ sufficiently large the procedure is $(\epsilon, \delta)$-DP with $\delta = \delta^* + O(1/n)$ and
\begin{equation}
\label{delta*_laplace_shuffled}
\delta^* = \frac{1 - e^{-C_1\exp(\epsilon/2)}}{C_1e^{\frac{\epsilon}{2}}}
\end{equation}
\end{theorem}

\begin{proof}
We use the result in \Cref{thm_shuffled_pnsgd} combined with (\ref{A_B_laplace}), and get that
\begin{equation*}
\delta = \frac{\l(1 - e^{\frac{\epsilon}{2} - \frac{L}{v(n)}}\r)_+ \cdot \l[1 - \l(1 - e^{\frac \epsilon2 - \frac{M(b-a)}{2\eta v(n)}}\r)_+^n\r]}{n\cdot e^{\frac \epsilon2 - \frac{M(b-a)}{2\eta v(n)}}}
\end{equation*}
Once we plug in the $v(n)$ defined in (\ref{v_laplace_shuffled}) we have that, when $n$ is sufficiently large,
\begin{align*}
\delta &= \frac{\Big(1 - e^{\frac{\epsilon}{2} - \frac{2L\eta \log(\frac{n}{C_1} + C_2)}{M(b-a)}}\Big)_+ \Big[1 - \Big(1 - \frac{C_1e^{\frac{\epsilon}{2}}}{n+C_1C_2}\Big)_+^{n}\Big]}{n\cdot e^{\frac \epsilon2 - \log(\frac{n}{C_1} + C_2)}} \\[3pt]
&= \frac{\l[1 - \l(1 - \frac{C_1e^{\frac{\epsilon}{2}}}{n+C_1C_2}\r)^{n}\r]}{n\cdot \frac{C_1e^{\frac{\epsilon}{2}}}{n+C_1C_2}}\cdot\left(1 + O\l(\frac{1}{n}\r)\right) \\[3pt]
&= \frac{1 - e^{-C_1\exp(\epsilon/2)}}{C_1e^{\frac{\epsilon}{2}}} + O\l(\frac{1}{n}\r)
\end{align*}

\end{proof}

The convergence result in \Cref{thm_shuffled_fixed_laplace} is confirmed by Figure \ref{fig_delta_shuffled_fixed_laplace}. In the left plot we see that $\delta$ converges to the $\delta^*$ defined in (\ref{delta*_laplace_shuffled}), while in the right plot we observe that the convergence rate is indeed $1/n$.

\subsection{Gaussian Noise}
\label{subsec_asymptotic_gaussian}

\begin{figure*}[!htbp]
  \centering
  \begin{minipage}{0.45\linewidth}
  \centering
  \includegraphics[width=\linewidth]{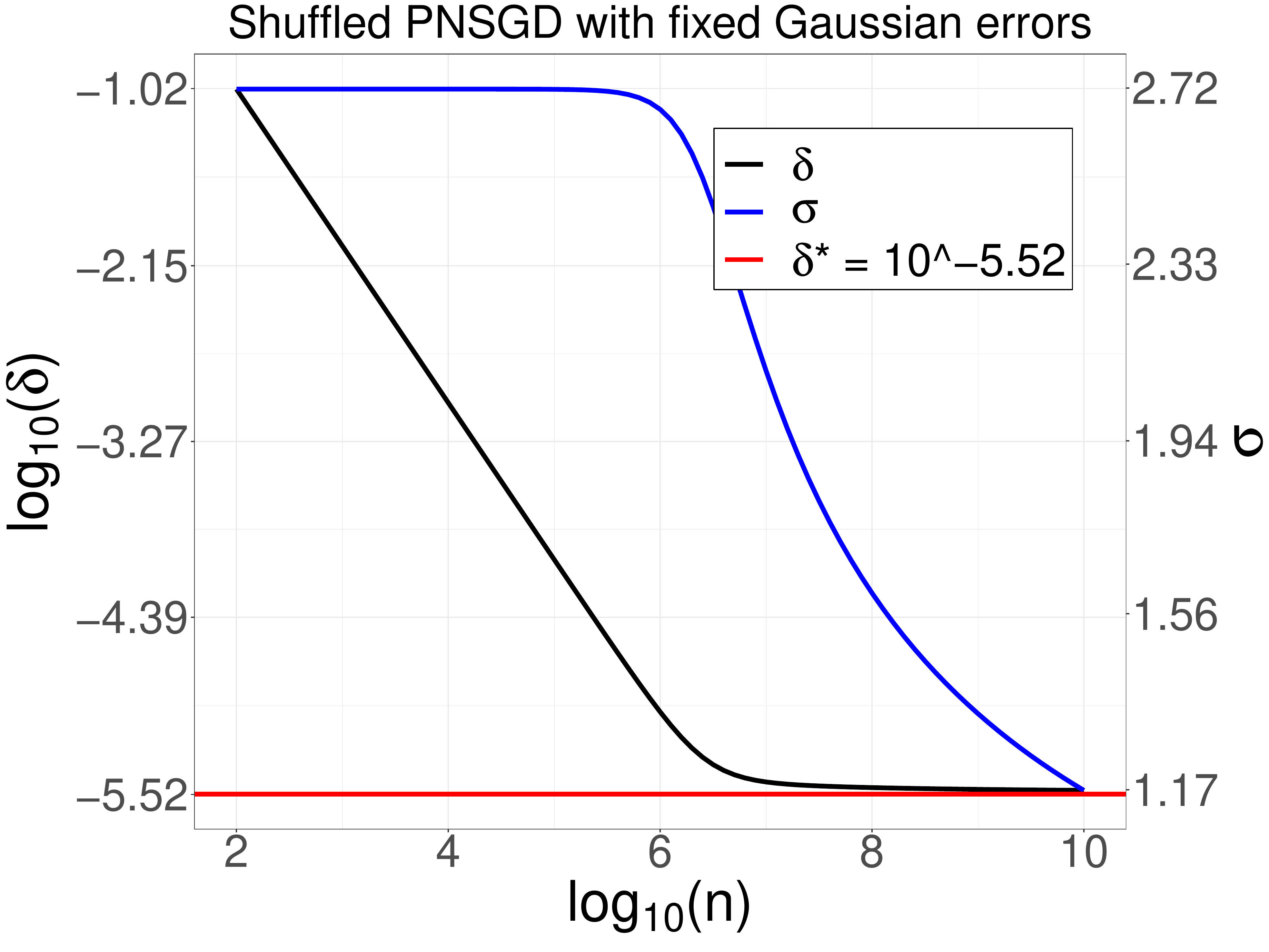}
  \vspace{2mm}
  \end{minipage}%
  \begin{minipage}{0.45\linewidth}
  \centering
  \includegraphics[width=\linewidth]{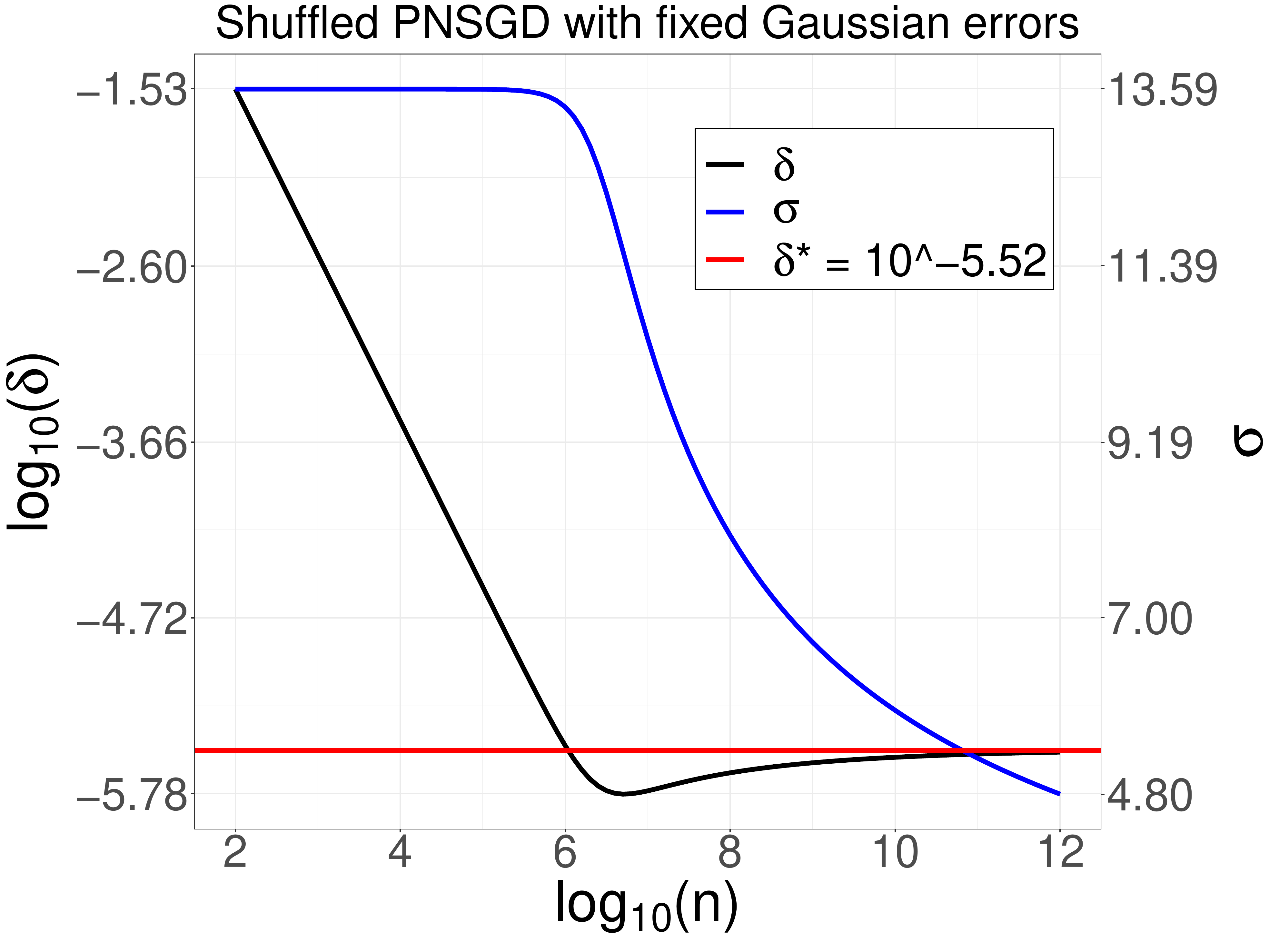}
  \vspace{2mm}
  \end{minipage}%
  \hfill
  \begin{minipage}{0.45\linewidth}
  \centering
  \includegraphics[width=\linewidth]{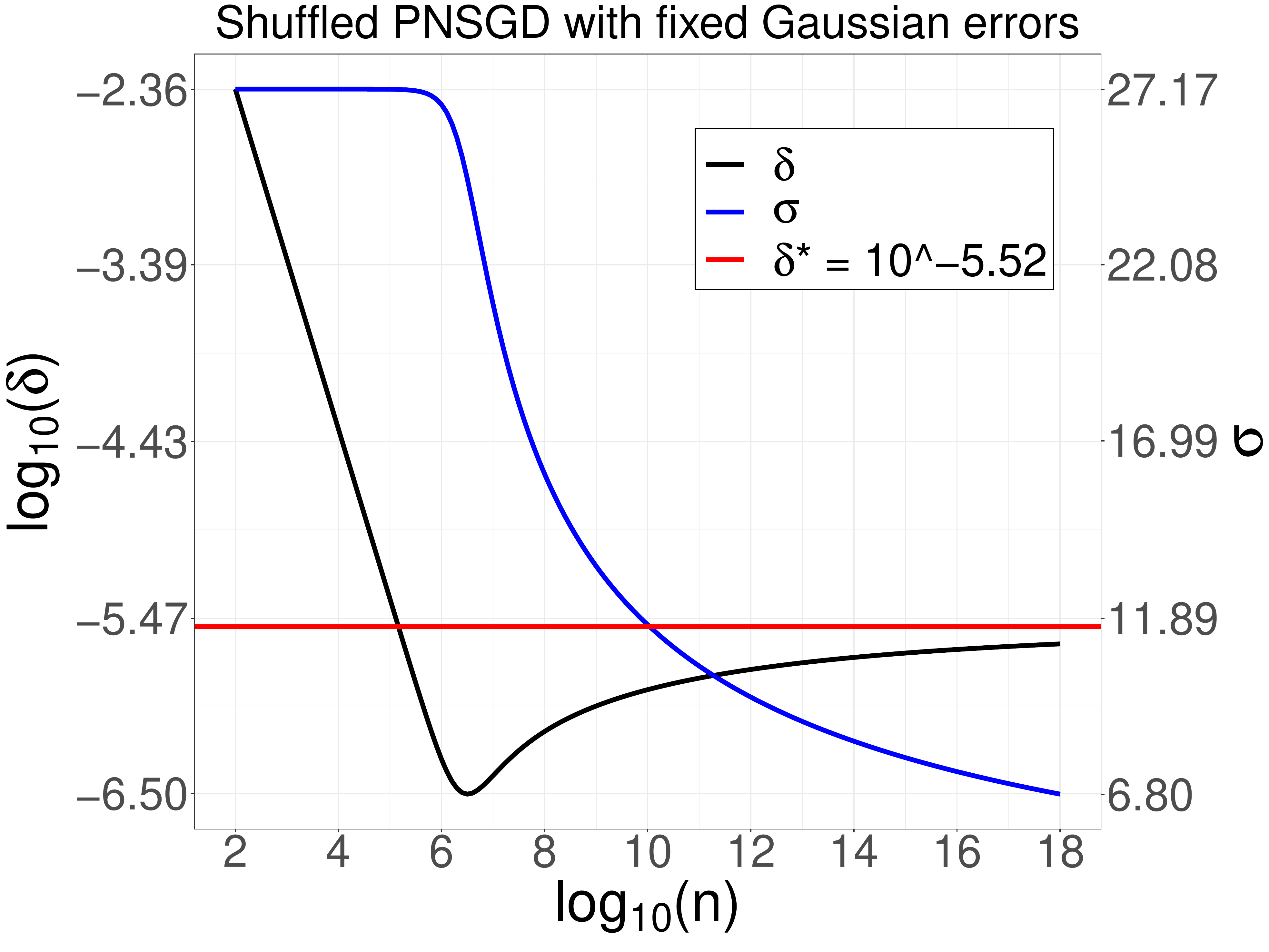}
  \end{minipage}%
  \begin{minipage}{0.45\linewidth}
   \centering
  \includegraphics[width=\linewidth]{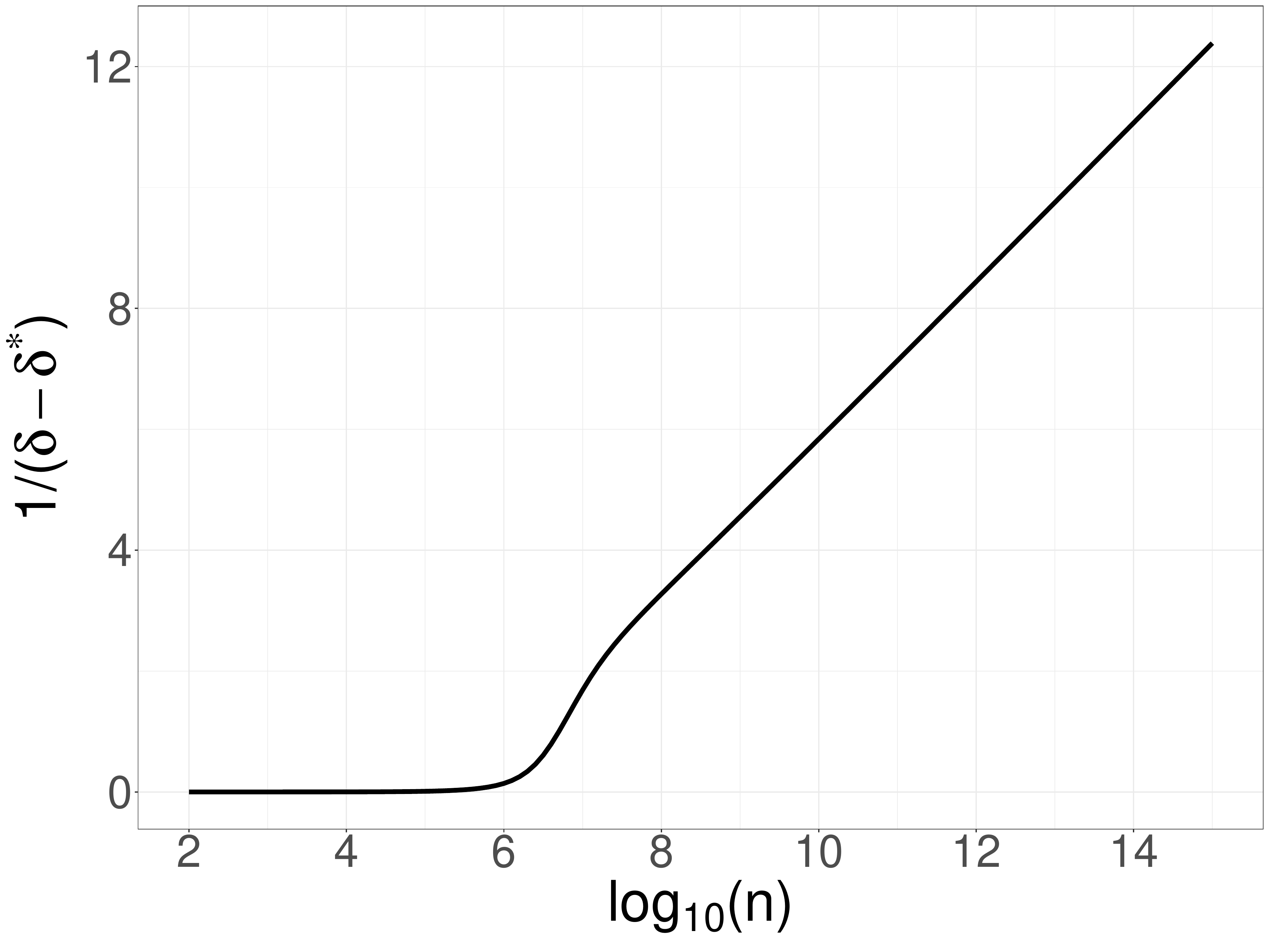}
  \end{minipage}
  \caption{Convergence of $\delta$ to $\delta^*$ defined in (\ref{delta_gaussian_shuffled}). We report in black the behavior of $\delta$ and in blue that of $\sigma(n)$ in (\ref{sigma_gaussian_shuffled}). We consider $\eta \in \{0.1, 0.02, 0.01\}$ and the other parameters are $L=10, \beta=0.5, \rho=0, \epsilon=1, D_{\K}=1, C_1=10^5$ and $C_2=100$. In the right panel we show that the convergence rate is $1/\log(n)$.}
\label{fig_delta_shuffled_fixed_gaussian}
\end{figure*}

Similarly to what we just proved in \Cref{subsec_asymptotic_laplace} we now discuss a result for the shuffled PNSGD with Gaussian noise $N(0, \sigma^2(n))$.

\begin{theorem}
\label{thm_shuffled_fixed_gaussian}
Consider the shuffled PNSGD algorithm with Gaussian noise $N(0, \sigma^2(n))$ which is fixed for each update, where
\begin{equation}
\label{sigma_gaussian_shuffled}
\sigma(n) = \frac{M D_\mathbb{K}}{2\eta \sqrt{W\l(\frac{n^2}{2 C_1^2 \pi} + C_2\r)}}
\end{equation}
and $W$ is the Lambert W function. Then, for $n$ sufficiently large, the procedure is $(\epsilon, \delta)$-DP with $\delta = \delta^* + O\l(\frac{1}{\log(n)}\r)$ and
\begin{equation}
\label{delta_gaussian_shuffled}
\delta^* = \frac{1 - e^{-2C_1e^{\frac{\epsilon}{2}}}}{2C_1e^{\frac{\epsilon}{2}}}
\end{equation}
\end{theorem}

Just like $v(n)$, the decay of the standard deviation $\sigma(n)$ is regulated by the parameters $C_1$ and $C_2$. The difference here is that, instead of a simple logarithmic decay, we now have a decay rate that depends on the Lambert W function, which is slightly harder to study analytically than the logarithm. Even though the Lambert W function is fundamentally equivalent to a logarithm when its argument grows, the difference with the Laplace case is also evident in the fact that the convergence of $\delta$ to $\delta^*$ happens more slowly, at a rate of $1/\log(n)$. 
The proof of the theorem is in \Cref{appendix_proof_gaussian_fixed}, and makes use of the following Lemma, which is proved in \Cref{appendix_proof_lemma}.

\begin{lemma}
\label{lem:simplify}
For $\theta_\gamma(r)$ defined in (\ref{definition_function_theta}), a sufficiently small $\sigma$ and two constants $c$ and $\epsilon$, we have
\begin{equation*}
\theta_{e^\epsilon}\l(\frac{c}{\sigma}\r)=1-\frac{1}{\sqrt{2\pi}}e^{\frac{\epsilon}{2}}e^{-\frac{c^2}{8\sigma^2}}\l(\frac{4\sigma}{c}+O(\sigma^3)\r).
\end{equation*}
\end{lemma}

The behavior described in \Cref{thm_shuffled_fixed_gaussian} is confirmed by what we see in \Cref{fig_delta_shuffled_fixed_gaussian}, where we can also observe that there are different patterns of convergence for $\delta$, both from above and from below the $\delta^*$ defined in (\ref{delta_gaussian_shuffled}). In the right-bottom plot we also see a confirmation that the convergence rate is the one we expected, since $(\delta-\delta^*)^{-1}$ increase linearly with respect to $\log(n)$ when $n$ is sufficiently large (notice that the y-axis is rescaled by a factor $10^6$).

\section{Multiple Epochs Composition}

We now consider a simple yet important extension of the result in \Cref{thm_shuffled_pnsgd}, where the shuffled PNSGD is applied for multiple epochs. In real experiments, e.g. when training deep neural networks, usually multiple passes over the data are necessary to learn the model. In such scenario, the updates are not kept secret for the whole duration of the training, but are instead released at the end of each epoch. The result proved in \Cref{thm_shuffled_pnsgd} states that for each epoch the procedure is $(\epsilon,\delta)$-DP with $\delta \leq A\cdot (1 - B^n)/[n(1-B)]$. We can then easily combine these privacy bounds using state-of-the-art composition tools, such as the Moments Accountant \citep{abadi2016deep}, 
$f$-DP and Gaussian DP \citep{dong2019gaussian}. We present some popular ways to compute the privacy loss after $E$ epochs. 

At the high level, we migrate from $(\epsilon,\delta)$ in DP to other regimes, Gaussian DP or R\'enyi DP, at the first epoch. Then we compose in those specific regimes until the end of training procedure. At last, we map from the other regimes back to $(\epsilon,\delta)$-DP.

\textbf{$f$-DP and Gaussian DP:} At the first epoch, we compute the initial $(\epsilon,\delta)$ and derive the four-segment curve $f_{\epsilon,\delta}$ for the type I/II errors trade-off (see Equation (5) and Proposition 2.5 in \cite{dong2019gaussian}). Then by Theorem 3.2 in \cite{dong2019gaussian}, we can numerically compose this trade-off function with Fourier transform for $E$ times, which can be accelerated by repeated squaring.
When the noise is Gaussian, we can alternatively use $\mu$ in GDP to characterize the trade-off function (i.e. the mechanism is $\mu$-GDP after the first epoch). Next, we apply Corollary 3.3 in \cite{dong2019gaussian} to conclude that the mechanism is $\sqrt{E}\mu$-GDP in the end. We can compute the final $(\epsilon,\delta)$ reversely from GDP by Corollary 2.13 in \cite{dong2019gaussian}.

\textbf{Moments Accountant:} Moments Accountant is closely related to R\'enyi DP (RDP), which composes easily: at the first epoch, we compute the $(\epsilon,\delta)$ of our PNSGD. By Proposition 3 in \cite{mironov2017renyi}, we can transfer from $(\epsilon,\delta)$-DP to $(\alpha,\epsilon+\frac{\log \delta}{\alpha-1})$ RDP. After the first epoch, the initial RDP can be composed iteratively by Moments Accountant\footnote{See \url{https://github.com/tensorflow/privacy/blob/master/tensorflow_privacy/privacy/analysis/rdp_accountant.py}}. The final $(\alpha',\epsilon')$ RDP is then mapped back to $(\epsilon,\delta)$-DP with $\epsilon=\epsilon'-\frac{\log \delta}{\alpha'-1}$.


\section{Online Results for Decaying Noises}
\label{sec_online_errors}

We now go back to the original framework of \citet{asoodeh2020privacy} and consider the PNSGD algorithm applied to the non-shuffled dataset. This time, however, we want to apply a different level of noise for each update, and see if we can get a convergence result for $\delta$ when $n\to\infty$. 
We then need to consider values of $A$ and $B$ in (\ref{A_B_gaussian}) and (\ref{A_B_laplace}) that depend on the specific index, and the privacy bound for the PNSGD with non-fixed noises and neighboring datasets that differ on index $i$ becomes
\begin{equation}
\label{delta_PNSGD_decay_errors}
\delta = A_i \cdot \prod_{t=i+1}^n B_t
\end{equation}
Here the definition of $A_i$ and $B_i$ is the same as in (\ref{A_B_gaussian}) and (\ref{A_B_laplace}) but the noise level $v$ and $\sigma$ is now dependent on the position of each element in the dataset.
In this scenario we can actually imagine adding new data to the dataset in an online fashion, without having to restart the procedure to recalibrate the noise level used for the first entries. It is clear that, in order to get convergence, the decay of the injected noise should be faster than in \Cref{thm_shuffled_fixed_laplace} and \Cref{thm_shuffled_fixed_gaussian}, since now the early entries receive an amount of noise that does not vanish as $n$ becomes large. However it is interesting to notice that for both the Laplace and Gaussian noise the only difference needed with the decay rate for $v(n)$ and $\sigma(n)$ defined before is an exponent $\alpha > 1$.


\subsection{Laplace Noise}
\label{subsec_online_Laplace}

We prove here the online result for the PNSGD with Laplace noise that decays for each entry. As anticipated, the decay is no longer the same for all entries and proportional to $1/\log(n)$ but now for the entry with index $j$ we have a decay which is proportional to $1/\log(j^\alpha)$.

\begin{figure}[!htbp]
    \centering
  \begin{minipage}{0.5\linewidth}
  \centering
  \includegraphics[width=\linewidth]{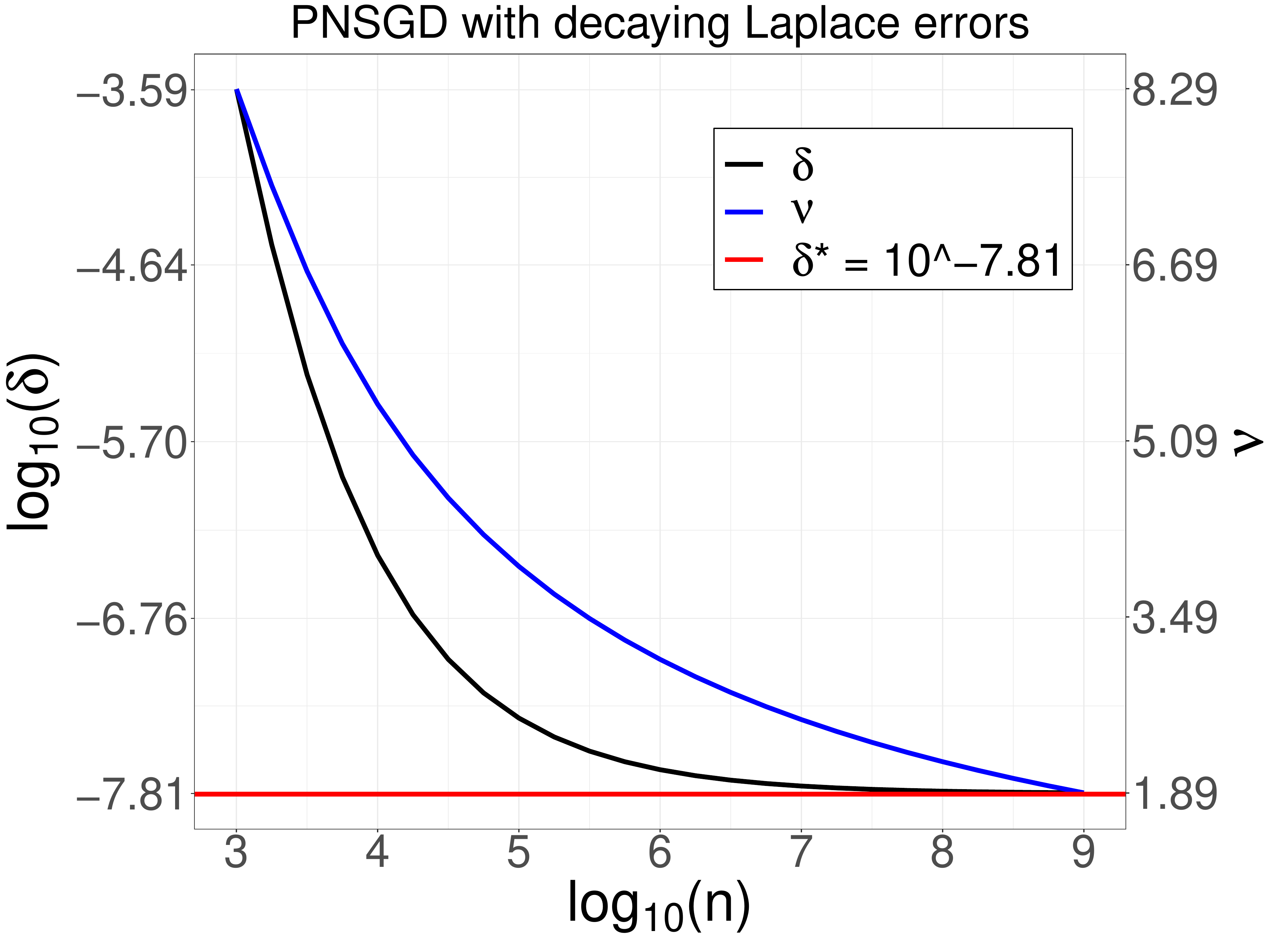}
  \end{minipage}%
  \hfill
  \begin{minipage}{0.49\linewidth}
  \centering
  \includegraphics[width=\linewidth]{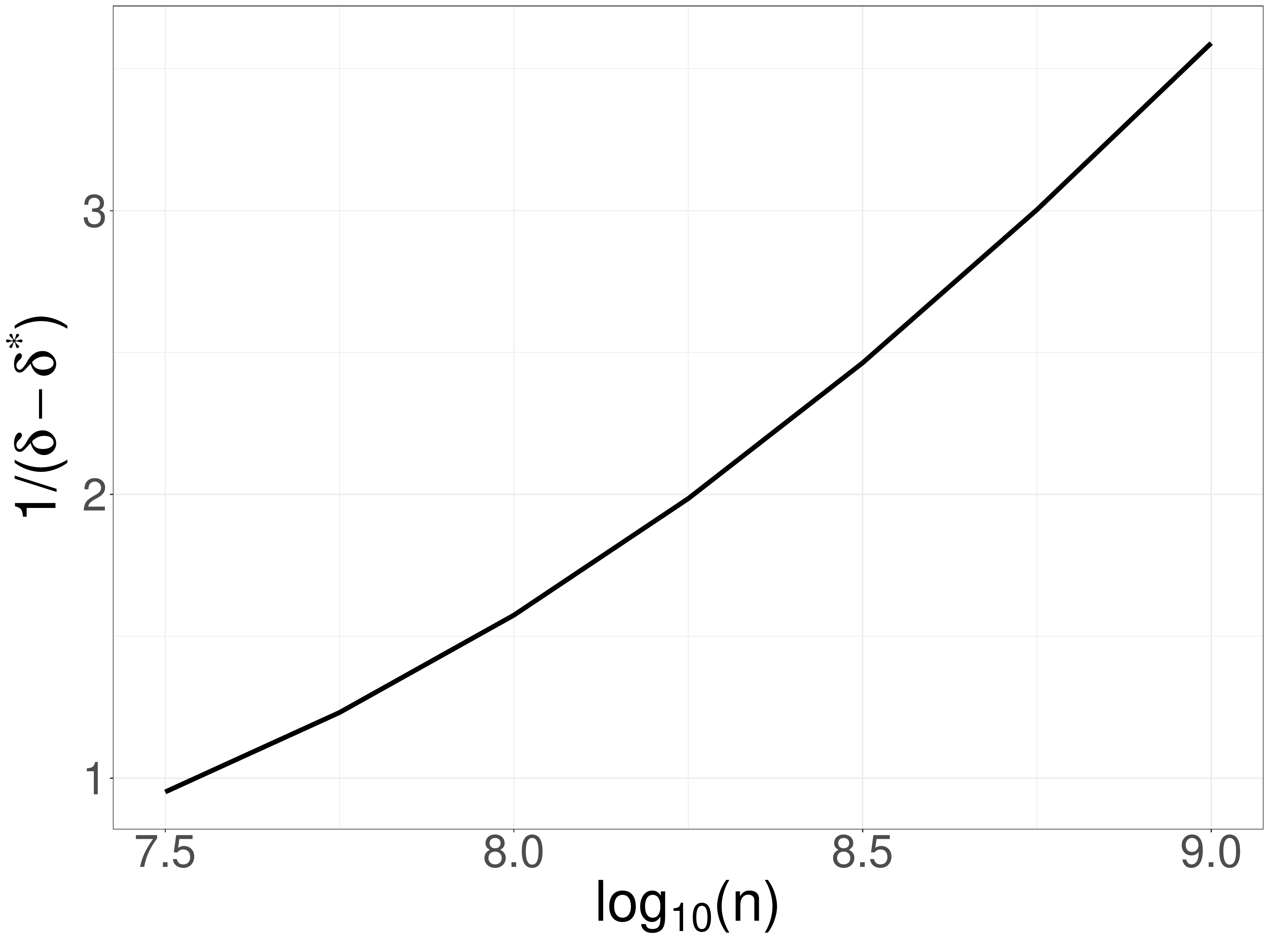}
  \end{minipage}
  \caption{(left) Convergence of $\delta$ to $\delta^*$ defined in (\ref{delta_convergence_laplace_decrease_v}). We report in black the behavior of $\delta$ and in blue that of $v_n$ defined in (\ref{v_laplace_decay}). The parameters considered are $L=10, \beta=0.5, \rho=0, \epsilon=1, \eta = 0.01, \alpha = 1.5, (a,b)=(0,1), i = 100, C_1=100$ and $C_2=100$. (right) The convergence rate is approximately $1/\log(n)$.}
\label{fig_delta_online_laplace}
\end{figure}

\begin{theorem}
\label{thm_online_laplace}
Consider the PNSGD where for update $j$ we use Laplace noise $\L(0, v_j)$, and 
\begin{equation}
\label{v_laplace_decay}
v_j = \frac{M(b-a)}{2\eta \log\l(j^\alpha/C_1 + C_2\r)}
\end{equation}
for $\alpha > 1$. Then as $n \to \infty$ the procedure is $(\epsilon, \delta^*)$-DP where
\vspace{-0.2cm}
\begin{equation}
\label{delta_convergence_laplace_decrease_v}
\delta^* =  \Big(1 - e^{\frac{\epsilon}{2} - \frac{2L\eta\log\l(i^\alpha/C_1 + C_2\r)}{M(b-a)}}\Big)_+ \exp\l\{{\int_{i+1}^\infty } \log\l( 1 - \frac{C_1 e^{\frac \epsilon2}}{x^\alpha + C_1 C_2}\r) dx\r\}
\end{equation}
and $i$ is the index where the neighboring datasets differ.
\end{theorem}
\begin{proof}
We show again that $\delta$ converges to a non-zero value as $n$ goes to $\infty$. In fact, again following the proof of (\citet{asoodeh2020privacy} Theorem 3), we get that,
\vspace{-0.2cm}
\begin{align*}
\delta &= \l(1 - e^{\frac{\epsilon}{2} - \frac{L}{v_i}}\r)_+ \cdot \prod_{t = i+1}^{n} \l(1 - e^{\frac \epsilon2 - \frac{M(b-a)}{2\eta v_t}}\r)_+ \\
&=  \Big(1 - e^{\frac{\epsilon}{2} - \frac{2L\eta\log(\frac{i^\alpha}{C_1} + C_2)}{M(b-a)}}\Big)_+ \prod_{t = i+1}^{n} \l(1 - \frac{C_1 e^{\frac \epsilon2}}{t^\alpha + C_1 C_2}\r)_+
\vspace{-0.3cm}
\end{align*}

We know that, for a sequence $a_t$ of positive values, $\prod_{t=1}^\infty (1 - a_t)$ converges to a non-zero number if and only if $\sum_{t=1}^\infty a_t$ converges. Here we have that 
$$\sum_{t = i+1}^{\infty} \frac{C_1 e^{\frac \epsilon2}}{t^\alpha + C_1 C_2} \leq \sum_{t = i+1}^{\infty} \frac{C_1 e^{\frac \epsilon2}}{t^\alpha}$$
and, since $\alpha > 1$ the right hand side converges, hence $\delta$ converges to a non-zero number. 
Let now $f(n) = \prod_{t = i+1}^{n} \l(1 - \frac{C_1 e^{\frac \epsilon2}}{t^\alpha + C_1 C_2}\r)_+$. To find the limit $f(\infty)$ we can first log-transform this function, and then upper bound the infinite sum with an integral before transforming back. Since $\log\l( 1 - \frac{C_1 e^{\frac \epsilon2}}{t^\alpha + C_1 C_2}\r)$ is monotonically increasing in $t$, we have
\begin{align*}
\log(f(n)) &= \sum_{t = i+1}^{n} \log\l( 1 - \frac{C_1 e^{\frac \epsilon2}}{t^\alpha + C_1 C_2}\r) \\
&< \int_{i+1}^n \log\l( 1 - \frac{C_1 e^{\frac \epsilon2}}{t^\alpha + C_1 C_2}\r)dt \\
&\to \int_{i+1}^\infty \log\l( 1 - \frac{C_1 e^{\frac \epsilon2}}{t^\alpha + C_1 C_2}\r)dt.
\end{align*}
This integral can be written in closed form using the hypergeometric function, or approximated numerically. 
\end{proof}

The convergence result that we get is slightly conservative, since $\delta^*$ in \Cref{delta_convergence_laplace_decrease_v} is an upper bound. However, following the previous proof, we can find an easy lower bound by just noticing that $\log(f(\infty)) >  \int_{i}^\infty \log\l( 1 - \frac{C_1 e^{\frac \epsilon2}}{t^\alpha + C_1 C_2}\r)dt$. When $i$ is not too small, the difference between the upper and lower bound is negligible, as it is confirmed by what we see in the left plot of \Cref{fig_delta_online_laplace}, where the convergence to the upper bound appears to be impeccable. Since the convergence is not exactly to $\delta^*$, we cannot find an explicit convergence rate the same way we did in \Cref{sec_asymptotic}. However, we see in the right plot of \Cref{fig_delta_online_laplace} that the convergence rate empirically appears to be $1/\log(n)$.

\subsection{Gaussian Noise}
\label{subsec_online_Gaussian}

When working with the Gaussian noises, the cumbersome form of the functions in (\ref{A_B_gaussian}) does not prevent us from finding a closed form solution for the limit $\delta^*$. Just as in the Laplace case we can find a conservative upper bound for $\delta^*$ which is very close to the true limit, as confirmed by the left plot of \Cref{fig_delta_online_gaussian}. Just as before, we notice again empirically from the right plot of \Cref{fig_delta_online_gaussian} that the convergence rate is $1/\log(n)$.

\begin{figure}[!htbp]
    \centering
  \begin{minipage}{0.49\linewidth}
  \centering
  \includegraphics[width=\linewidth]{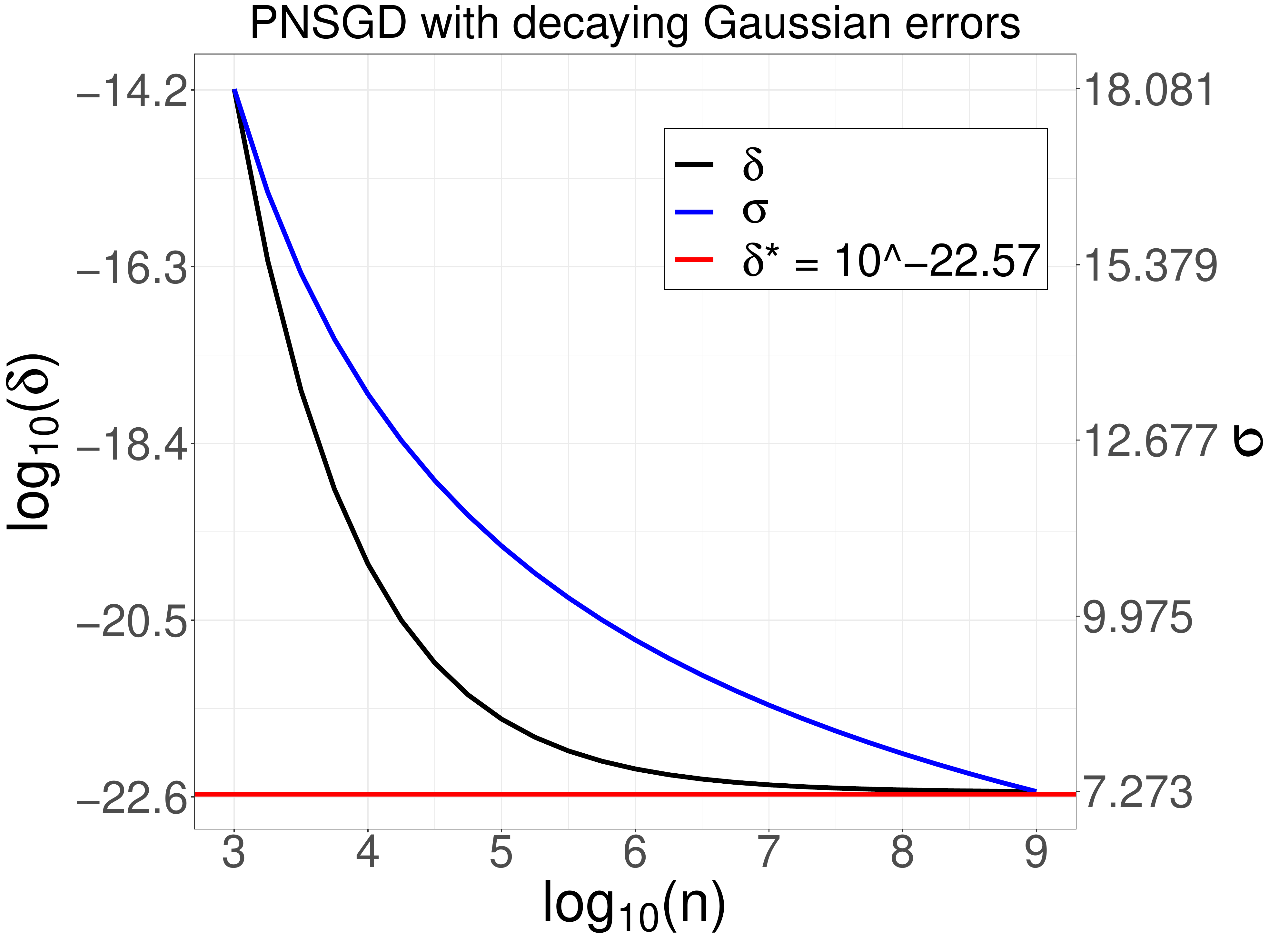}
  \end{minipage}%
  \hfill
  \begin{minipage}{0.49\linewidth}
  \centering
  \includegraphics[width=\linewidth]{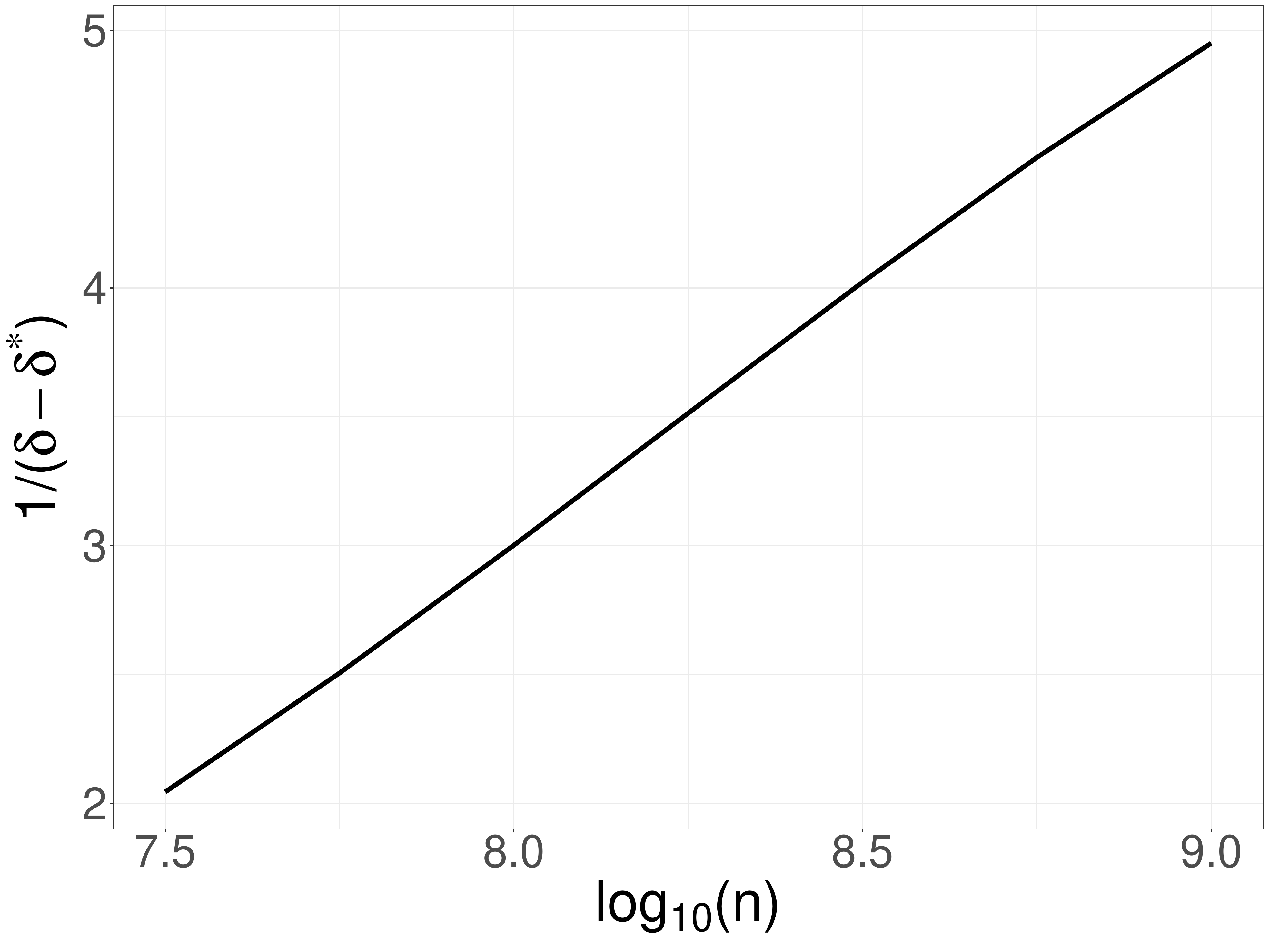}
  \end{minipage}
  \caption{(left) Convergence of $\delta$ to $\delta^*$ defined in (\ref{delta_convergence_gaussian_decrease_sigma}). We report in black the behavior of $\delta$ and in blue that of $\sigma_n$ defined in (\ref{sigma_gaussian_decay}). The parameters considered are $L=10, \beta=0.5, \rho=0, \epsilon=1, \eta = 0.01, \alpha = 1.5, D_\K=1, i = 100, C_1=100$ and $C_2=100$. (right) The convergence rate is approximately $1/\log(n)$.}
\label{fig_delta_online_gaussian}
\end{figure}


\begin{theorem}
\label{thm_online_gaussian}
Consider the PNSGD where for update $j$ we use Gaussian noise $N(0, \sigma^2_j)$, and 
\vspace{-0.2cm}
\begin{equation}
\label{sigma_gaussian_decay}
\sigma_j = \frac{MD_\K}{2\eta \sqrt{W\l(\frac{j^{2\alpha}}{2\pi C_1^2} + C_2\r)}}
\vspace{-0.3cm}
\end{equation}
for $\alpha > 1$. Then as $n \to \infty$ the procedure is $(\epsilon, \delta^*)$-DP where
\vspace{-0.2cm}
\begin{equation}
\label{delta_convergence_gaussian_decrease_sigma}
\delta^* =  \theta_{e^\epsilon}\Big(\frac{2L}{\sigma_i}\Big) \exp\l\{\int_{i+1}^\infty  \log\l(\theta_{e^\epsilon}\l(2 \sqrt{W\l(\frac{x^{2\alpha}}{2\pi C_1^2} + C_2\r)}\r)\r) dx \r\}
\end{equation}
and $i$ is the index where the neighboring datasets differ.
\end{theorem}

The proof of this result is in \Cref{appendix_proof_gaussian_online}, and makes use again of Lemma \ref{lem:simplify} to show that asymptotically the terms $B_t$ in (\ref{delta_PNSGD_decay_errors}) behave approximately as $1 - O(1/t^\alpha)$, so that convergence is guaranteed for the same reason as in \Cref{thm_online_laplace}.

\section{Conclusion}

In this work, we have studied the setting of privacy amplification by iteration in the formulation proposed by \citet{asoodeh2020privacy}, and proved that their analysis of PNSGD also applies to the case where the data are shuffled first. This is a much more common practice than the randomly-stopped PNSGD, originally proposed, because of a clear advantage in terms of accuracy of the algorithm.
We proved two asymptotic results on the decay rate of noises that we can use, either the Laplace or the Gaussian injected noise, in order to have asymptotic convergence to a non-trivial privacy bound when the size of the dataset grows.
We then showed that these practical bounds can be combined using standard tools from the composition literature.
Finally we also showed two result, again for Laplace or Gaussian noise, that can be obtained in an online setting when the noise does not have to be recalibrated for the whole dataset but just decayed for the new data.

\bibliographystyle{apalike}
\bibliography{my_bibliography}

\newpage
\appendix

\section{Proof of Lemma \ref{lem:simplify}}
\label{appendix_proof_lemma}

Recall from the definition \eqref{definition_function_theta}:
\begin{align}
\theta_{e^\epsilon}\l(\frac{c}{\sigma}\r) = Q\l(\frac{\epsilon\sigma}{c} - \frac c{2\sigma}\r) - e^\epsilon Q\l(\frac{\epsilon\sigma}{c} + \frac c{2\sigma}\r)
\label{eq:theta_temp}
\end{align}
We apply the following approximation of the normal cumulative density function, valid for large positive $x$,
\begin{align}
Q(x):=\frac{1}{\sqrt{2\pi}}\int_x^{\infty} e^{-\frac{u^2}{2}} du = \frac{1}{\sqrt{2\pi}}e^{-\frac{x^2}{2}} \l(\frac{1}{x} - \frac{1}{x^3} + \frac{3}{x^5} + \cdots\r) = \frac{1}{\sqrt{2\pi}}e^{-\frac{x^2}{2}} \l(\frac{1}{x} + O\l(\frac{1}{x^3}\r)\r)
\end{align}
and similarly, for large negative values of $x$
\begin{align*}
Q(x):=\frac{1}{\sqrt{2\pi}}\int_x^{\infty} e^{-\frac{u^2}{2}} du = 1 +\frac{1}{\sqrt{2\pi}}e^{-\frac{x^2}{2}}\l(\frac{1}{x} + O\l(\frac{1}{x^3}\r)\r) .
\end{align*}
Therefore \eqref{eq:theta_temp} can be reformulated as
\begin{align*}
\theta_{e^\epsilon}\l(\frac{c}{\sigma}\r) 
&= 1+\frac{1}{\sqrt{2\pi}}e^{-\frac 12\l(\frac{\epsilon^2\sigma^2}{c^2}+\frac{c^2}{4\sigma^2}\r)+\frac\epsilon2}\l(\frac{1}{\frac{\epsilon\sigma}{c}-\frac{c}{2\sigma}}+O\l(\frac{1}{\frac{\epsilon\sigma}{c}-\frac{c}{2\sigma}}\r)^3\r)
\\
&\quad
-\frac{1}{\sqrt{2\pi}}e^{\epsilon}e^{-\frac 12\l(\frac{\epsilon^2\sigma^2}{c^2}+\frac{c^2}{4\sigma^2}\r)-\frac\epsilon2}\l(\frac{1}{\frac{\epsilon\sigma}{c}+\frac{c}{2\sigma}}+O\l(\frac{1}{\frac{\epsilon\sigma}{c}+\frac{c}{2\sigma}}\r)^3\r)
\\
&=1-\frac{1}{\sqrt{2\pi}}e^{\frac{\epsilon}{2}}e^{-\frac 12\l(\frac{\epsilon^2\sigma^2}{c^2}+\frac{c^2}{4\sigma^2}\r)}\l(\frac{4\sigma}{c}+O(\sigma^3)\r)
\\
&=1-\frac{1}{\sqrt{2\pi}}e^{\frac{\epsilon}{2}}e^{-\frac{c^2}{8\sigma^2}}\l(\frac{4\sigma}{c}+O(\sigma^3)\r)
\end{align*}

\section{Proof of \Cref{thm_shuffled_fixed_gaussian}}
\label{appendix_proof_gaussian_fixed}

From Theorem \ref{thm_shuffled_pnsgd} we know that
\begin{equation}
	\label{delta_gaussian_error_implicit}
	\delta = \frac{\theta_{e^\epsilon}\l(\frac{2L}{\sigma(n)}\r) \cdot \l[1 - \theta_{e^\epsilon}\l(\frac{M D_{\mathbb{K}}}{\eta \sigma(n)}\r)^{n}\r]}{n \cdot \l[1 -  \theta_{e^\epsilon}\l(\frac{M D_{\mathbb{K}}}{\eta \sigma(n)}\r)\r]}
\end{equation}
We show that with $\sigma(n)$ that decays according to (\ref{sigma_gaussian_shuffled}) we have that 
$$\theta_{e^\epsilon}\l(\frac{2L}{\sigma(n)}\r)\to 1 \qquad\text{and}\qquad \theta_{e^\epsilon}\l(\frac{M D_{\mathbb{K}}}{\eta \sigma(n)}\r) \to 1-\frac{2C_1e^\frac{\epsilon}{2}}{n}.$$
Let's first focus briefly on the behavior of the Lambert W function. Formally, the Lambert W function is an implicit function defined as the inverse of $f(w) = w e^w$, meaning that for any $x$ one has $W(x)e^{W(x)}=x$. As an interesting fact, we note that the Lambert W function's behavior is approximately logarithmic, e.g. $\log(x) > W(x) > \log_4(x)$, where by $\log$ we denote the natural logarithm. We also denote the argument of the W Lambert function in $\sigma(n)$ as $x=\frac{n^2}{2 C_1^2 \pi} + C_2$.
Using this fact, an immediate consequence of Lemma \ref{lem:simplify} is that, when plugging in the $\sigma(n)$ from \eqref{sigma_gaussian_shuffled}, we get
$$\theta_{e^\epsilon}\l(\frac{2L}{\sigma(n)}\r)=1-o(\sigma^3)=1-o\l(\frac{1}{\sqrt{W^3(x)}}\r)=1-o\l(\frac{1}{(\log n)^{3/2}}\r)$$
since $e^{-\frac{c^2}{8\sigma^2}}\cdot\frac{1}{\sigma^2}\to 0$ as the exponential decays faster than the polynomial.
Next, we study $\theta_{e^\epsilon}\l(\frac{M D_{\mathbb{K}}}{\eta \sigma(n)}\r)$.
Again by Lemma \ref{lem:simplify}, we have
\begin{align}{\label{align_theta_gaussian_thm4}} \nonumber
\theta_{e^\epsilon}\l(\frac{M D_{\mathbb{K}}}{\eta \sigma(n)}\r)
&=1-\frac{1}{\sqrt{2\pi}}e^{\frac{\epsilon}{2}}e^{-\frac{M^2D_{\K}^2}{8\eta^2\sigma(n)^2}}\l(\frac{4\eta\sigma(n)}{MD_{\K}} + O(\sigma(n)^3)\r)
\\ \nonumber
&=1-\frac{1}{\sqrt{2\pi}}e^{\frac{\epsilon}{2}}e^{-\frac{W(x)}{2}}\l(\frac{2}{\sqrt{W(x)}} + O\l(\frac{1}{W(x)^{3/2}}\r)\r)
\\ \nonumber
&=1-\frac{2 e^{\frac{\epsilon}{2}}}{\sqrt{2\pi}} \frac{1}{\sqrt{e^{W(x)}W(x)}} + O\l(\frac{1}{\sqrt{e^{W(x)} W(x)^3}}\r)
\\ \nonumber
&=1-\frac{2 e^{\frac{\epsilon}{2}}}{\sqrt{2\pi x}} + O\l(\frac{1}{\sqrt{x} \log(x)}\r)
\\
&= 1 - \frac{2 C_1 e^{\frac{\epsilon}{2}}}{n} + O\l(\frac{1}{n \log(n)}\r)
\end{align}
Going back to the expression in (\ref{delta_gaussian_error_implicit}) we finally have that
\begin{align*}
\delta &= \frac{\l(1 - o\l(\frac{1}{(\log(n))^{3/2}}\r)\r)\l[1 - \l(1 - \frac{2C_1e^\frac{\epsilon}{2}}{n} + O\l(\frac{1}{n \log(n)}\r)\r)^{n}\r]}{n \cdot \l[1 - \l(1 - \frac{2 C_1 e^\frac{\epsilon}{2}}{n} + O\l(\frac{1}{n \log(n)}\r)\r)\r]} \\[3pt]
&=\frac{\l(1 - o\l(\frac{1}{(\log(n))^{3/2}}\r)\r)\l[1 - \l(1 - \frac{2C_1e^\frac{\epsilon}{2}+O(1/\log(n))}{n}\r)^{n}\r]}{2 C_1 e^\frac{\epsilon}{2} + O\l(\frac{1}{\log(n)}\r)}
\\
&=\frac{\l(1 - o\l(\frac{1}{(\log(n))^{3/2}}\r)\r)\l[1-e^{-2C_1e^\frac{\epsilon}{2}+O(1/\log(n))}\r]}{2 C_1 e^\frac{\epsilon}{2} + O\l(\frac{1}{\log(n)}\r)}
\\
&=\l(1 - o\l(\frac{1}{(\log(n))^{3/2}}\r)\r)\l[1-e^{-2C_1e^\frac{\epsilon}{2}}+O\l(\frac{1}{\log(n)}\r)\r]\l(\frac{1}{2C_1e^{\frac{\epsilon}{2}}}-O\l(\frac{1}{\log(n)}\r)\r)
\\
&= \frac{1 - e^{-2C_1e^\frac{\epsilon}{2}}}{2C_1e^\frac{\epsilon}{2}} + O\l(\frac{1}{\log(n)}\r)
\end{align*}
where the last equality holds because $ f(n)=O\l(\frac{1}{\log(n)}\r) + o\l(\frac{1}{(\log(n))^{3/2}}\r) = O\l(\frac{1}{\log(n)}\r)$.

\section{Proof of \Cref{thm_online_gaussian}}
\label{appendix_proof_gaussian_online}

This proof combines elements of the proofs of \Cref{thm_shuffled_fixed_gaussian} and \Cref{thm_online_laplace}.
We start by studying the behavior of $\theta_{e^\epsilon}\l(\frac{M D_{\mathbb{K}}}{\eta \sigma_t}\r)$ as $t$ grows. We define $x = \frac{t^{2\alpha}}{2\pi C_1^2} + C_2$ so that $\sigma_t = \frac{MD_\K}{2\eta \sqrt{W(x)}}$ and get, as in (\ref{align_theta_gaussian_thm4}),
\begin{align*}
\theta_{e^\epsilon}\l(\frac{M D_{\mathbb{K}}}{\eta \sigma_t}\r) &= 1-\frac{1}{\sqrt{2\pi}}e^{\frac{\epsilon}{2}}e^{-\frac{M^2D_{\K}^2}{8\eta^2\sigma_t^2}}\l(\frac{4\eta\sigma_t}{MD_{\K}} + O(\sigma_t^3)\r) \\
&= 1-\frac{2 e^{\frac{\epsilon}{2}}}{\sqrt{2\pi x}} + O\l(\frac{1}{\sqrt{x} \log(x)}\r) \\
&= 1 - \frac{2 C_1 e^{\frac{\epsilon}{2}}}{t^\alpha} + O\l(\frac{1}{t^\alpha \log(t)}\r)
\end{align*}
This already confirms us that $\delta^*$ converges to a finite non zero value, since the asymptotic behavior of each term in the infinite product is the same as in the Laplace case. To express such limit in a more tractable way we follow the proof of \Cref{thm_online_laplace} and write $f(n) = \prod_{t=i+1}^n \theta_{e^\epsilon}\l(\frac{M D_{\mathbb{K}}}{\eta \sigma_t}\r)$ and approximate the infinite sum $\log(f(\infty))$ with an integral.

\begin{align*}
\log(f(n)) &= \sum_{t = i+1}^{n} \log\l(\theta_{e^\epsilon}\l(\frac{M D_{\mathbb{K}}}{\eta \sigma_t}\r)\r) \\
&= \sum_{t = i+1}^{n} \log\l(\theta_{e^\epsilon}\l(2 \sqrt{W\l(\frac{t^{2\alpha}}{2\pi C_1^2} + C_2\r)}\r)\r) \\
&< \int_{i+1}^n  \log\l(\theta_{e^\epsilon}\l(2 \sqrt{W\l(\frac{x^{2\alpha}}{2\pi C_1^2} + C_2\r)}\r)\r) dx \\
&\to \int_{i+1}^\infty  \log\l(\theta_{e^\epsilon}\l(2 \sqrt{W\l(\frac{x^{2\alpha}}{2\pi C_1^2} + C_2\r)}\r)\r) dx \\[3pt]
\end{align*}
This confirms us that 
\begin{equation*}
\delta^* =  \theta_{e^\epsilon}\Big(\frac{2L}{\sigma_i}\Big) \cdot \exp\l\{\int_{i+1}^\infty  \log\l(\theta_{e^\epsilon}\l(2 \sqrt{W\l(\frac{x^{2\alpha}}{2\pi C_1^2} + C_2\r)}\r)\r) dx \r\} .
\end{equation*}

\end{document}